\documentclass[journal,10pt,draftclsnofoot,onecolumn]{IEEEtran}
\usepackage{cite}
\ifCLASSINFOpdf
\else
\fi
\usepackage[cmex10]{amsmath}
\usepackage{amssymb}
\newtheorem{theorem}{Theorem}
\newtheorem{lemma}{Lemma}
\newtheorem{definition}{Definition}
\newtheorem{conjecture}{Conjecture}
\newtheorem{proposition}{Proposition}
\newtheorem{corollary}{Corollary}

\newtheorem{remark}{Remark}

\begin{document}

\title{On the Bounds of certain Maximal Linear Codes in a Projective Space}
\author{Srikanth~Pai,
        B.~Sundar~Rajan,~\IEEEmembership{Fellow,~IEEE}% <-this % stops a space
\thanks{The authors are with the Department
of Electrical and Communication Engineering, Indian Institute of Science, Bangalore-560012, India. E-mail: srikanthbpai@gmail.com, bsrajan@ece.iisc.ernet.in.}}
%\author{\IEEEauthorblockN{Srikanth Pai and B. Sundar Rajan}
%\author{Srikanth B. Pai and B. Sundar Rajan\\
%Dept. of ECE, Indian Institute of Science,\\
%Bangalore 560012, India\\
%Email:\{spai,bsrajan\}@ece.iisc.ernet.in}
%\IEEEauthorblockA{Dept. of ECE, Indian Institute of Science,\\
%Bangalore 560012, India\\
%Email:\{spai,bsrajan\}@ece.iisc.ernet.in}}

\maketitle

\IEEEpeerreviewmaketitle

\begin{abstract}
The set of all subspaces of $\mathbb{F}_q^n$ is denoted by $\mathbb{P}_q(n)$. The subspace distance $d_S(X,Y) = \dim(X)+ \dim(Y) - 2\dim(X \cap Y)$ defined on $\mathbb{P}_q(n)$ turns it into a natural coding space for error correction in random network coding. A subset of $\mathbb{P}_q(n)$ is called a code and the subspaces that belong to the code are called codewords. Motivated by classical coding theory, a linear coding structure can be imposed on a subset of $\mathbb{P}_q(n)$. Braun, Etzion and Vardy conjectured that the largest cardinality of a linear code, that contains $\mathbb{F}_q^n$, is $2^n$. In this paper, we prove this conjecture and characterize the maximal linear codes that contain $\mathbb{F}_q^n$.
\end{abstract}

\section{Introduction}

%ADD the definition of subspace distance
%ADD the definition of projective space
%\cite{BEV} for most part
% state tht it follows from context what 0 stands for.

\IEEEPARstart{T}{he} class of all subspaces of $\mathbb{F}_q^n$ is denoted by $\mathbb{P}_q(n)$ and will be called a {\it projective space}. The projective space $\mathbb{P}_q(n)$ is endowed with a subspace distance. Given two subspaces $X$ and $Y$, $$\text{d}_S(X,Y) := \dim(X)+ \dim(Y) - 2\dim(X \cap Y).$$ The metric arises as a distance function in the field of error correction in random network coding. In this context, a subset of the projective space is called a { \it code} and the elements that belong to a code are called {\it codewords}. Koetter and Kschischang showed that one can view transmission of messages over networks as transmission of elements of a code in a projective space\cite{KoeKschi}. In their formulation, the amount of `error' between a transmitted subspace and a received subspace is measured by their subspace distance. They showed that one can correct a fixed amount of errors in the network given that all the subspaces in the code are pairwise sufficiently far apart. This approach mimicked the theory of classical error correcting codes. And therefore the theory of error correction in random network coding draws some motivations from the field of classical error correcting codes.

Classical error correcting codes are constructed over $\mathbb{F}_2^n$. Here a `code' is an arbitrary subset of $\mathbb{F}_2^n$ and the `codewords' are $n$ length vectors whose co-ordinates are elements of $\mathbb{F}_2$. The distance, called the {\it Hamming distance}, between two codewords is measured by the number of differing co-ordinates between the two vectors. The space of $\mathbb{F}_2^n$ which is equipped with a Hamming distance is called the {\it Hamming space}. Most of the theory on classical error correcting codes focus on a special class of codes called {\it linear codes}. Linear codes are just subspaces of $\mathbb{F}_2^n$. A subspace, by definition, is an abelian group under vector addition. And it turns out that the Hamming distance is translation invariant with respect to this addition. The geometry of the Hamming distance combines with group structure to give a rich theory of classical binary error correcting codes \cite{MacSlo}. 

However there are a few differences between error correction coding over projective spaces and Hamming spaces. If the volume of a set is defined as the number of elements in the set whenever the set is finite, one can study how the volume of a ball changes as the centre is translated. In Euclidean spaces and Hamming spaces, the volume of a ball is independent of its centre. But the volume of a ball is not independent of its centre in a projective space. A technical way of restating this observation is that the geometry induced by the subspace distance $d_S$ on $\mathbb{P}_q(n)$ is not distance-regular (unlike the Hamming space) and one cannot mimic the tricks of sphere-covering/packing bounds for the Hamming space \cite{KoeKschi}. And therefore the usual geometric intuitions fail in projective spaces. However, it turns out that restricting the projective spaces to a class of subspaces with a particular dimension gets this distance regular property. Therefore it becomes interesting to consider the set of all subspaces of particular dimension. We shall use the notation $\mathbb{P}_q(n,k)$ to denote the set of all elements of $\mathbb{P}_q(n)$ with dimension $k$.The authors of \cite{KoeKschi} lay the groundwork, derives bounds and construct a few codes on $\mathbb{P}_q(n,k)$. These codes are called {\it constant dimension codes} and recently there was a lot of activity in this area \cite{TraManRos}, \cite{EtzSilb}, \cite{GadYan}. 

On the other hand, our focus is on non-constant dimension codes. Even though the distance-regular property fails, one can recover a form of Gilbert-Varshamov bound for non-constant dimension codes \cite{EtzVar} and derive a form of Singleton bound for the projective space \cite{SriBSR}. The concepts of `linearity' and `complements' in projective spaces are examined in \cite{BEV}. The definitions are motivated from their binary error correction coding counterpart. The authors discuss different possibilities of complements and linearity and they state a few conjectures. This paper is mainly focused on resolving one conjecture concerning the linearity. 

A code in $\mathbb{P}_q(n)$ which imitates the group structure of a binary vector space and has the translation invariant property on the subspace distance is called a {\it linear code in the projective space} in \cite{BEV}. Henceforth, we will just call it a linear code. One can construct such a linear code with $2^n$ codewords by identifying a basis as the generators of the group. A linear code constructed in this fashion is called {\it a code derived from a fixed basis}. It follows from \cite{BEV} that $\mathbb{P}_q(n)$ cannot be a linear code. So a maximal linear code is somewhere between these two extremes. The main conjecture claims that the largest possible cardinality of a maximal linear code must be {\bf equal} to the cardinality of a code derived from a fixed basis. In other words, it states that the largest linear code in a projective space has a cardinality of $2^n$. As a first step towards proving this result, a special case of a linear code over $\mathbb{P}_q(n)$ that contains $\mathbb{F}_q^n$ is considered. It is conjectured that a linear code that contains $\mathbb{F}_q^n$ as a codeword can contain a maximum of $\binom{n}{k}$ $k$-dimensional codewords \cite{BEV}. In this paper, we prove this conjecture. We further show that maximal linear codes that contains $\mathbb{F}_q^n$ as a codeword must be derived from a fixed basis. It is important to note that if we drop the condition of `$\mathbb{F}_q^n$ is a codeword', one can construct linear codes that are not derived from a fixed basis (See Section $3$, Example $1$ in \cite{BEV}).

The organization of our paper is as follows: In Section \ref{Overview} we present the conjecture formally and state all the important definitions relating to the conjecture. The proof of the conjecture and a few additional results are presented in Section \ref{Proofs}. We conclude the paper with general remarks about the conjecture and a few open problems in Section \ref{Conclusion}.

{\it Notations}: $\mathbb{F}_q^n$ represents the $n$-dimensional vector space over the finite field $\mathbb{F}_q$. The class of all subspaces of $\mathbb{F}_q^n$ is denoted by $\mathbb{P}_q(n)$. $\mathbb{P}_q(n,k)$ represents the set of all elements of $\mathbb{P}_q(n)$ with dimension $k$. The $\text{span}(S)$ for a subset $S$ of $\mathbb{F}_q^n$ stands for the linear span of the elements of $S$. Occasionally, we identify the one-dimensional space $\text{span}(v) \in \mathbb{P}_q(n)$ with the vector $v \in \mathbb{F}_q^n$. The trivial space and the null element of the vector space are both represented by $0$ and they are meant to be distinguished from context. Given two sets $X$ and $Y$, $X \triangle Y:= (X \cup Y) \setminus (X \cap Y)$ denotes the symmetric difference of $X$ and $Y$. $A \oplus B$ represents the direct sum of two disjoint subspaces $A$ and $B$.

\section{Overview of the Conjecture}
\label{Overview}
Based on motivations to mimic the classical linear code structure on $\mathbb{F}_2^n$ endowed with the Hamming distance, a linear code on $\mathbb{P}_q(n)$ is defined in \cite{BEV} as follows: 

\begin{definition}
A code ${\cal C} \subseteq \mathbb{P}_q(n)$ is a linear code if 
\begin{enumerate}
\item there exists a binary operation $\boxplus$ on $\cal C$, such that $({\cal C}, \boxplus)$ is an Abelian group.
\item The additive identity of $({\cal C}, \boxplus)$ is the trivial space $0$.
\item For every $X \in {\cal C}$, we have $X \boxplus X = 0$.
\item $\text{d}_S(X,Y) = \text{d}_S(X \boxplus Z, Y \boxplus Z)$ for all $X,Y,Z \in {\cal C}$.
\end{enumerate}
\end{definition}

A couple of remarks on the definition of a linear code in $\mathbb{P}_q(n)$:
\begin{remark}
The first three conditions of the definition allow us to view a linear code ${\cal  C}$ as a subspace of a vector space of characteristic $2$. The third property states that every element is an {\it idempotent}.
\end{remark}

\begin{remark}
The last property in the definition is called the {\it translation invariant property}. In words, this property states that the subspace distance between $X$ and $Y$ does not change if both were `translated' by the same codeword $Z$.
\end{remark}

The following proposition is an immediate corollary of the definition of a linear code.
\begin{proposition}
\label{dimensionofsum}
If $X$ and $Y$ are two codewords of a linear code $({\cal C},\boxplus)$, then 
\[\dim(X \boxplus Y) = \text{d}_S(X,Y) .\]
\end{proposition}
\begin{proof}
Given two codewords $X$ and $Y$, using translation invariance and idempotence, we get:
\begin{eqnarray*}
 \text{d}_S(X,Y) &=& \text{d}_S(X \boxplus Y, Y \boxplus Y) \\
          &=& \text{d}_S(X \boxplus Y, 0) \\
					&=& \dim(X \boxplus Y) + \dim(0) - 2\dim((X\boxplus Y)\cap 0)\\
					&=& \dim(X \boxplus Y)
					\end{eqnarray*}
\end{proof}

A linear code can be defined on every projective space $\mathbb{P}_q(n)$ as described in \cite{BEV}. Pick a basis ${\cal B} = \{e_1,e_2,...,e_n\}$ for $\mathbb{F}_q^n$. Define the code ${\cal C}_{\cal B}$ as follows:

$${\cal C}_{\cal B} = \{V | V = \text{span}(S), S \subseteq {\cal B} \}$$

We can verify the simpler group properties. Given a set $S \subseteq {\cal B}$, define $V_S := \text{span}(S)$. Given two elements $V_X$ and $V_Y$ from ${\cal C}_{\cal B}$, define $$V_X \boxplus V_Y :=\text{span}(X \triangle Y).$$ The trivial space $0$ is spanned by the empty set. And since the symmetric difference of a set with itself is empty, the $\boxplus$ operation is an involution. The associativity and the translation invariance can be verified (See \cite{BEV}). The code ${\cal C}_{\cal B}$ has $2^n$ codewords; one for each subset of $\cal B$. A linear code of this form will be called `a code derived from a fixed basis'. Clearly any code constructed like this is generated from the basis of $n$ one dimensional spaces. However, we will see in Theorem \ref{one-dim space} that {\it merely} containing $n$ one dimensional spaces is sufficient to guarantee that the code is of the above type. We record our formal definition for `a code derived from a fixed basis' below. 

\begin{definition}
A linear code ${\cal C}$ in $\mathbb{P}_q(n)$ is {\it derived from a fixed basis} if the group ${\cal C}$ is generated by $n$ one-dimensional spaces in ${\cal C}$.
\end{definition}

As it will be seen in Corollary \ref{fullspacenotlinear}, one cannot define a group operation to make the entire space $\mathbb{P}_q(n)$ a linear code. A natural question, therefore, is whether linear codes with cardinality larger than $2^n$ exist. The authors in \cite{BEV} claim that it does not. A special case of the conjecture asks if this conjecture can be proved for linear codes that contain $\mathbb{F}_q^n$. The actual conjecture just asks this for the case of $q=2$. We record the conjecture formally as follows:

\begin{conjecture}
If a linear code ${\cal C}$ over $\mathbb{P}_q(n)$ contains $\mathbb{F}_q^n$ as a codeword, then $|{\cal C}| \leq 2^n$. Specifically $$|{\cal C} \cap \mathbb{P}_q(n,k)| \leq \binom{n}{k}.$$
\end{conjecture}

\begin{remark}
 It should be noted that the conjecture and our proof works for any field! In other words, the largest abelian group on the subspaces of $\mathbb{F}^n$ that contains the full space (with the translation invariant property of the subspace metric) can have at most $2^n$ elements in the group. In the particular case, when $\mathbb{F}$ is infinite, there are infinitely many subspaces and yet our conditions on the group force its size to be at most $2^n$.
\end{remark}

We prove this conjecture in the next section. We also prove that the equality $|{\cal C}| = 2^n$ holds if and only if ${\cal C}$ is derived from a fixed basis.

\section{Proof of the Conjecture}
\label{Proofs}

In this section, we will formally prove the conjecture through a series of smaller results. The following two lemmas are also proved in \cite{BEV}. And we record it here, because we will use them in the proof of the main theorem.

\begin{lemma}
If $X,Y,Z$ belong to a linear code ${\cal C}$ and $X \boxplus Y = Z$, then $Z \boxplus Y = X$ and $Z \boxplus X = Y$.
\label{cyclic sum}
\end{lemma}
\begin{proof}
This lemma follows from the property $A \boxplus A = 0$ for all $A \in {\cal C}$. $$ Z \boxplus Y = (X \boxplus Y) \boxplus Y = X \boxplus (Y \boxplus Y) = X.$$ One can similarly prove $Z \boxplus X = Y.$
\end{proof}

The following lemma proves that adding two disjoint subspaces in the group gives the direct sum of the two subspaces. 
\begin{lemma}
If $X, Y$ belong to a linear code ${\cal C}$ and $X \cap Y = 0$, then $X \boxplus Y = X \oplus Y.$
\label{disjoint sum}
\end{lemma}
\begin{proof}
Suppose $Z = X \boxplus Y$, then from Proposition \ref{dimensionofsum} and $X \cap Y=0$, we have:
$$\dim(Z) = \dim(X) + \dim(Y)$$

From Lemma \ref{cyclic sum}, $Y = Z \boxplus X$ and from Proposition \ref{dimensionofsum},
\begin{eqnarray*}
\dim(Y) &=& \dim(Z) + \dim(X) - 2\dim(Z \cap X).\\
\dim(Y) &=& \dim(Y) + 2 (\dim(X) - \dim(Z\cap X)) .
\end{eqnarray*}

Therefore, $\dim(X) = \dim(Z\cap X)$ and thus $X \subseteq Z$. Similarly we can show that $Y \subseteq Z$. Thus $X+Y \subseteq Z$ which along with $\dim(Z) = \dim(X) + \dim(Y)$ proves that $Z = X \oplus Y$.
\end{proof}

The next lemma shows that given a codeword that is a subset of another codeword, the addition of these two codewords `splits' the larger codeword. In other words, the larger codeword can be generated from its smaller components.

\begin{lemma}
If $X, Y$ belong to a linear code ${\cal C}$ and $X \subset Y$, then there exists a $Z \in {\cal C}$ such that $Y = X \oplus Z.$
\label{splitting}
\end{lemma}
\begin{proof}
Let $Z := X \boxplus Y$. The dimension of $Z$ can be calculated:
\begin{eqnarray*}
\dim(Z) &=& \dim(Y) + \dim(X) - 2\dim(Y \cap X)\\
\dim(Z) &=& \dim(Y) - \dim(X)
\end{eqnarray*}

From Lemma \ref{cyclic sum}, $Y = Z \boxplus X$. The above calculation of $\dim(Z)$ gives,
\begin{eqnarray*}
\dim(Y) &=& \dim(Z) + \dim(X) - 2\dim(Z \cap X)\\
\dim(Y) &=& \dim(Y) - 2\dim(Z\cap X).
\end{eqnarray*}

Therefore $\dim(Z\cap X) = 0$ and thus $Z \cap X = 0$. From Lemma \ref{disjoint sum}, $Y = X \boxplus Z = X \oplus Z.$

\end{proof}

Lemma \ref{splitting} can be used to decompose larger dimensional codewords using the one-dimensional subspaces contained in the larger subspace. In the following lemma, we will first show that the one-dimensional subspaces in a linear code must be linearly independent. Linear dependence of one dimensional spaces will contradict the uniqueness of a inverse. In the second part of the lemma, we will show that the subspaces spanned by one dimensional codewords belong to the code. 

Let $\cal O$ denote the set of all one dimensional spaces in ${\cal C}$. Assume $|{\cal O}| = m$ and let $V_i \in {\cal O}$ represent distinct one dimensional spaces spanned by $v_i \in \mathbb{F}_q^n$ respectively for $1 \leq i \leq m$. When we say the subspaces $\{V_i\}$s are linearly independent, it means that the corresponding $\{v_i\}$s are linearly independent.

\begin{proposition}
Let $\cal O$ denote the set of all one dimensional codewords in ${\cal C}$. Then,
\begin{enumerate}
\item $\cal O$ is linearly independent.
\item $S \subseteq {\cal O} \Rightarrow \text{span}(S) \in {\cal C}.$

\end{enumerate}

\label{subsetgensubspace}
\end{proposition}

\begin{proof}

Proof of $(1)$: The proof is by contradiction. Without loss of generality, let us assume that $D = \{v_1,v_2, ..., v_k\}$ is a minimally dependent set of vectors. Since $V_i$s are distinct one dimensional spaces, we can assume $k > 2$. By minimality of dependence, it follows that $$\text{span}(D) = \text{span}\{v_1,v_2, ..., v_{k-1}\} = \text{span}\{v_2,v_3, ..., v_k\}.$$ From Lemma \ref{disjoint sum}, we see that $$V_1 \boxplus V_2 ....\boxplus V_{k-1} = \text{span}(D) = V_2 \boxplus V_3 ....\boxplus V_{k}.$$ But by cancellation of $V_2 \boxplus V_3 ....\boxplus V_{k-1}$ on both sides, we get $V_1 = V_k$ which contradicts the distinctness of the one dimensional spaces. Therefore $\cal O$ is linearly independent. 

Proof of $(2)$: Let $S \subseteq {\cal O}$. We prove the claim by induction on $|S|$. The case of $|S| = 1$ is easily verified. Partition $S$ as $S = T \cup \{V\}$, then $\text{span}(T) \in {\cal C}$ by induction hypothesis. From part $(1)$ of this proposition, it follows that $V$ is linearly independent of $T$ and therefore $\text{span}(T) \cap V = 0$. Applying Lemma \ref{disjoint sum}, we get $$\text{span}(S) = \text{span}(T) \oplus V = \text{span}(T) \boxplus V \in {\cal C}.$$ By closure property of ${\cal C}$, $\text{span}(S) \in {\cal C}.$

\end{proof}

We shall now show that the group operation for every code derived from a fixed basis is the same. We will specifically show that if $A$ and $B$ are two code words from such a code, then there exists a basis ${\cal O}$, such that $A \boxplus B$ is actually the span of the symmetric difference of some two subsets of $\cal O$. 

\begin{proposition}
Let ${\cal C}$ be a code derived from a fixed basis, and $\cal O$ be the set of one dimensional codewords that generate $\cal C$. Then for any $A, B \in {\cal C}$, there exists $X, Y \subseteq {\cal O}$ such that
$$A \boxplus B = \text{span}(X \triangle Y).$$  
\end{proposition}

\begin{proof}
Define $$X:=A \cap {\cal O}, \quad Y:=B \cap {\cal O}.$$ Since $X \subseteq A$, we have $\text{span}(X) \subseteq A$. We claim that $\text{span}(X) = A$. If not, $A \boxplus \text{span}(X) = P \neq 0$. From Lemma \ref{splitting}, $P \subset A$ and $A = P \oplus \text{span}(X).$ But ${\cal O}$ generates ${\cal C}$ and therefore there exists a one dimensional space $\{p\} \in {\cal O}$ that belongs to $P$ as well. So $\{p\}$ belongs to $A \cap {\cal O}$ but not to $X$, contradicting $X = A \cap {\cal O}$. Thus we have $\text{span}(X) = A$ and similarly, we have $\text{span}(Y) = B$. Note that from Lemma \ref{disjoint sum}, the group addition of all the linearly independent one dimensional spaces belonging to $X$ is simply the span of $X$. In the following expressions, the $\sum$ is with respect to the group operation. Therefore 
\begin{align*} 
A &= \sum_{E \in X \setminus X \cap Y} E \boxplus \sum_{F \in X \cap Y} F \\
B &= \sum_{D \in Y \setminus X \cap Y} D \boxplus \sum_{F \in X \cap Y} F.
\end{align*}

Using the fact $F \boxplus F = 0$, we get: $$A \boxplus B = \sum_{E \in X \setminus X \cap Y} E \boxplus \sum_{D \in Y \setminus X \cap Y} D .$$ Again using Lemma \ref{disjoint sum} and the fact that $X, Y, Z$ contain one-dimensional subspaces, \begin{eqnarray*}
A \boxplus B &=& \text{span}\left(\left(X \setminus (X \cap Y)\right) \cup \left(Y \setminus (X \cap Y)\right)\right) \\
&=& \text{span}\left(\left(X \cup Y\right) \setminus (X \cap Y)\right) \\ 
&=& \text{span}\left( X \triangle Y \right).
\end{eqnarray*}
\end{proof}

We can choose a basis ${\cal B}$ and consider the linear code as subspaces spanned by the subsets of ${\cal B}$ where the group operation is defined by span of the symmetric difference of the generating sets. The above proposition shows that all `codes derived from a fixed basis' are of this form. We will now prove a simple sufficient condition for the code to be derived from a fixed basis. 

A consequence of Proposition \ref{subsetgensubspace} is that the subgroup generated by the one-dimensional codewords in a linear code is the collection of subspaces spanned by subsets of the one-dimensional spaces in $\cal C$. However, the following theorem proves that the subgroup is the entire group if the subgroup {\it merely contains} $n$ one dimensional spaces. Alternatively, the theorem proves that if a code has $n$ one dimensional spaces, then the code is entirely generated from it.

\begin{theorem}
\label{one-dim space}
If a linear code ${\cal C}$ in $\mathbb{P}_q(n)$ contains $n$ one-dimensional spaces then ${\cal C}$ is derived from a fixed basis.
\end{theorem}

\begin{proof}
If $\cal C$ contains $n$ one-dimensional spaces from $\mathbb{P}_q(n)$, then it follows, from the first part of Proposition \ref{subsetgensubspace}, that the one dimensional spaces form a basis ${\cal B}$ for $\mathbb{F}_q^n$. We have to prove that these $n$ one dimensional spaces generate the code. Let $S$ be the set of all subspaces in $\cal C$ that are spanned by a subset of ${\cal B}$. 

We want to prove that $S = {\cal C}$. 

The proof is by contradiction. Suppose ${\cal C} \setminus S$ is non-empty. Pick a minimal subspace $M$ in ${\cal C} \setminus S$. Consider the following set:
$${\cal M} = \{K \in S | K \cap M = 0\}$$
Due to minimality and Lemma \ref{splitting}, $M$ does not contain any one-dimensional space in $S$ and therefore ${\cal M}$ is non-empty. Let $K$ be a maximal subspace of ${\cal M}$. We claim that $$K + M = \mathbb{F}_q^n.$$ We prove this claim by contradiction. Suppose $K + M \subsetneq \mathbb{F}_q^n$, then there exists an element $v \in {\cal B}$ that does not belong to $K + M$. This means that $K \subsetneq \text{span}\{K,v\}$ and since $K$ is a maximal subspace that belongs to $\cal M$, $\text{span}\{K,v\} \notin \cal M$. But $\text{span}\{K,v\} \in S$ since $K \in S$ and $v \in {\cal B}$. Therefore we conclude $\text{span}\{K,v\} \cap M \neq 0$. Since $K \in S$, it must be spanned by a subset ${\cal K} \subseteq {\cal B}$. Let $0 \neq m \in \text{span}\{K,v\} \cap M$, then $$m = \sum_{k \in {\cal K}}c_k k + c_v v.$$ We note that $c_v$ is non zero since $K \cap M  = 0$. Rearranging the equation to express $v$, we get $$v = \sum_{k \in {\cal K}}-\frac{c_k}{c_v} k + \frac1{c_v} m.$$ But this contradicts $v \notin K + M.$ Therefore $K+M = \mathbb{F}_q^n.$

By Lemma \ref{disjoint sum}, $K \boxplus M = \mathbb{F}_q^n$. However $K_{\text{comp}} := \text{span}({\cal B} \setminus {\cal K})$ also satisfies $K \boxplus K_{\text{comp}} = \mathbb{F}_q^n$. Cancellation laws hold for $\boxplus$, and thus we have $M = K_{\text{comp}}$ which contradicts the fact that $M \notin S$ (since $K_{\text{comp}} \in S$). Hence ${\cal C} \subseteq S$ i.e. only subspaces that can be spanned by a subset of the chosen basis ${\cal B}$ can belong to the code. Since $S$ contains spans of subsets of ${\cal B}$, the second part of Proposition \ref{subsetgensubspace} ensures the reverse inclusion $S \subseteq {\cal C}$. Thus we have established $S = {\cal C}$.
\end{proof}

Theorem \ref{one-dim space} will be used to show that the equality holds in the conjecture only when the code is derived from a fixed basis. In other words, the equality of the conjecture holds only for codes that are designed by fixing a basis, then spanning all subsets of the fixed basis (where the codewords are the span of the subsets). Therefore, the symmetric diffence of generating sets is the group operation that achieves the maximum size of a linear code, given that the $\mathbb{F}_q^n$ belongs to the linear code. 

Theorem \ref{one-dim space} also shows that the entire projective space $\mathbb{P}_q(n)$ cannot be a linear code. The following corollary can also be deduced from older results \cite{BEV}.

\begin{corollary}
There does not exist a group operation $\boxplus$ such that $\mathbb{P}_q(n)$ is a linear code.
\label{fullspacenotlinear}
\end{corollary}
\begin{proof}
From Theorem \ref{one-dim space}, it follows that $\mathbb{P}_q(n)$ is a linear code derived from a fixed basis. However, the first part of Proposition \ref{subsetgensubspace} does not allow a linear code to have more than $n$ one dimensional spaces. Therefore $\mathbb{P}_q(n)$ cannot be a linear code.
\end{proof}

Now we will prove a few lemmas that will be useful to prove the main theorem. At the heart of the proof, we need a technical lemma due to Lovasz \cite{Lov} given below.

\begin{lemma}
\cite[Theorem 4.9]{Lov}
If  $A_1, A_2, ..., A_m $ are $r$-dimensional spaces and $B_1, B_2, ..., B_m$ are $s$-dimensional spaces with the property $A_i \cap B_j = 0 \Longleftrightarrow i=j$, then $$m \leq \binom{r+s}{s}.$$
\label{Lovasz}
\end{lemma}

Consider ${\cal C}_k := {\cal C} \cap \mathbb{P}_q(n,k)$ for all $1 \leq k \leq n$. The next lemma show that ${\cal C}_k$ and ${\cal C}_{n-k}$ satisfy the conditions of Lemma \ref{Lovasz}, namely the two classes of subspaces have the same number of elements and under some ordering they satisfy the intersection conditions. In the first part of the following lemma, we shall establish that $|{\cal C}_k| = |{\cal C}_{n-k}|$. We shall also show, in the second part, that ${\cal C}_k$ and ${\cal C}_{n-k}$ satisfy the intersection conditions specified in Lemma \ref{Lovasz}. 

\begin{lemma}

Define a map $\phi: {\cal C} \longrightarrow {\cal C}$ as $C \longmapsto \mathbb{F}_q^n \boxplus C$.  We claim the following:

\begin{enumerate}
\item $|{\cal C}_k| = |{\cal C}_{n-k}|$.
\item If $C \in {\cal C}_k$ and $D \in {\cal C}_{n-k}$ then $C \cap D = 0$ if and only if $\phi(C) = D$.
\end{enumerate}
\label{conditions}
\end{lemma}

\begin{proof}
The map is left multiplication by $\mathbb{F}_q^n$ and since ${\cal C}$ is closed under $\boxplus$, $\phi$ is well-defined.

Proof of $(1)$: We will show that $\phi(\phi(C)) = C$. 
\begin{eqnarray*}
 \phi(\phi(C)) &=& {\mathbb{F}_q^n} \boxplus ({\mathbb{F}_q^n} \boxplus C)       \nonumber \\
   &=& ({\mathbb{F}_q^n} \boxplus {\mathbb{F}_q^n}) \boxplus C \nonumber \\
   &=& 0 \boxplus C = C.
\end{eqnarray*}

Thus $\phi$ is an involution and therefore bijective. To show $|{\cal C}_k| = |{\cal C}_{n-k}|$ we prove that $\phi({\cal C}_k) = {\cal C}_{n-k}$. In other words, we will show that a $C \in {\cal C}_k$ gets mapped to a subspace of dimension $n-k$ by the following calculation:
 \begin{eqnarray*}
 \dim ({\cal C} \boxplus \mathbb{F}_q^n) &=& \dim ({\cal C}) + \dim({\mathbb{F}_q^n}) - 2 \dim({\cal C} \cap \mathbb{F}_q^n)       \nonumber \\
   &=& k + n - 2k \nonumber \\
   &=& n-k
\end{eqnarray*}

Therefore $\phi(C) \in {\cal C}_{n-k}$ and similarly $\phi({\cal C}_{n-k}) \subseteq {\cal C}_{k}$. Thus $\phi(\phi({\cal C}_{n-k})) \subseteq \phi({\cal C}_{k}) \subseteq {\cal C}_{n-k}$. Since $\phi$ is an involution, we know that $\phi(\phi({\cal C}_{n-k})) = {\cal C}_{n-k}$ and therefore it follows that $\phi({\cal C}_k) =  {\cal C}_{n-k}$.

Proof of $(2)$: Suppose $D = \phi(C) = \mathbb{F}_q^n \boxplus C$. Applying Lemma \ref{splitting}, we get $C \oplus D = \mathbb{F}_q^n$ which implies $C \cap D = 0$. For the converse, we note that $C \cap D = 0$ and $\dim(C) + \dim(D) = n$ gives us $C \oplus D = \mathbb{F}_q^n$ and therefore $D = \mathbb{F}_q^n \boxplus C = \phi(C)$.
\end{proof}

We finally have the necessary lemmas to prove the conjecture. The following theorem is the main theorem of this paper which proves the conjecture and characterizes the equality case.

\begin{theorem}
If ${\cal C}$ is a linear code over $\mathbb{P}_q(n)$ that contains $\mathbb{F}_q^n$, then $|{\cal C}| \leq 2^n$. Further the equality holds if and only if ${\cal C}$ is derived from a fixed basis.
\end{theorem}

\begin{proof}
The first part of Lemma \ref{conditions} shows that $|{\cal C}_k| = |{\cal C}_{n-k}|$ and along with second part of Lemma \ref{conditions}, ${\cal C}_k$ and ${\cal C}_{n-k}$ are collections of subspaces that satisfy the conditions of Lemma \ref{Lovasz}. Applying that lemma to ${\cal C}_k$ and ${\cal C}_{n-k}$ we get, 
$$|{\cal C}_{n-k}| = |{\cal C}_{k}| \leq \binom{k+(n-k)}{k} = \binom{n}{k}.$$

The theorem is finally proved by the following calculation:
$$|{\cal C}| = \sum_{k=0}^{n} |{\cal C}_{k}| \leq \sum_{k=0}^{n} \binom{n}{k} = 2^n.$$

We note that the equality holds only if $|{\cal C}_{k}| =  \binom{n}{k}$ for all $0 \leq k \leq n$. In particular, this means $|{\cal C}_{1}|=n$ and thus the code contains $n$ one dimensional subspaces. Applying Proposition \ref {one-dim space}, we infer that $\cal C$ is derived from a fixed basis.
\end{proof}

\begin{remark}
A natural question is whether every maximal linear code contains the full space. This would prove the conjecture that the size of a maximal linear code in $\mathbb{P}_q(n)$ is at most $2^n$. However, this strategy will fail because there are examples of maximal codes in the literature that do not contain the full space \cite{BEV}. %This also means that the condition $\mathbb{F}_q^n \in {\cal C}$ is necessary for the equality case of the main theorem.
\end{remark}

\section{Conclusion}
\label{Conclusion}

We have proved the conjecture on the bound of a linear code, in a projective space, that contains the full space as a codeword. The maximal linear codes containing the full space are characterized as codes derived from a fixed basis. However, all known examples of linear codes that does not contain the full space also have at most $2^n$ codewords. This general conjecture would be an interesting problem to consider for future research. 

Further it can be shown that these proofs can be adapted in a more general framework of lattices. We have not explored the various lattice connections in this paper. However Hamming spaces are examples of distributive lattices and projective spaces are examples of modular lattices. The codes that are derived from a fixed basis are embeddings of distributive lattices into the modular lattice of a projective space. We conjecture that the size of the largest distributive sub-lattice of a geometric modular lattice of height $n$ must be $2^n$. 

\section*{Acknowledgment}

The first author thanks Prashant N. for stimulating discussions about projective spaces and Gautam Shenoy for verifying the proofs of the assertions.

%\bibliography{mybib}
%\bibliographystyle{ieeetr}

\end{document}